\documentclass[10pt]{article}
\usepackage[latin1]{inputenc}
\usepackage{amsmath}
\usepackage{amssymb, amscd, amsthm}
\usepackage[all]{xy}
\usepackage[dvips]{graphicx}
\usepackage{verbatim}
\usepackage[perpage,symbol*]{footmisc}
\newtheorem{theorem}{Theorem}

\newtheorem{lemma}{\textbf{Lemma}}[section]
\newtheorem{remark}{\textbf{Remark}}[section]

\newtheorem{example}{\textbf{Example}}[section]

\usepackage{longtable}
\usepackage{multirow}
\usepackage{slashbox}

\newcommand{\F}{\mathbb{F}}

\begin{document}

\baselineskip 17pt
\title{\Large\bf New Quantum MDS Codes over Finite Fields}

\author{\large  Xiaolei Fang \quad\quad  Jinquan Luo\footnote{The authors are with School of Mathematics
and Statistics \& Hubei Key Laboratory of Mathematical Sciences, Central China Normal University, Wuhan China 430079.
\newline  E-mails: fangxiaolei@mails.ccnu.edu.cn(X.Fang), luojinquan@mail.ccnu.edu.cn(J.Luo).}}
\date{}
\maketitle

{\bf Abstract}: In this paper, we present three new classes of $q$-ary quantum MDS codes utilizing generalized Reed-Solomon codes satisfying Hermitian
self-orthogonal property. Among our constructions, the minimum distance of some $q$-ary quantum MDS codes can be bigger than $\frac{q}{2}+1$. Comparing to previous known constructions, the lengths of codes in our constructions are more flexible.



\section{Introduction}

\quad\; Quantum error-correcting codes play an important role in quantum information transmission and quantum computation. Due to the establishment of
the connections between quantum codes and classical codes (see [\ref{AK},\ref{CRSS},\ref{Shor}]), great progress has been made in the study of quantum
error-correcting codes. One of these connections shows that quantum codes can be constructed from classical linear error-correcting codes
satisfying symplectic, Euclidean or Hermitian self-orthogonal properties (see [\ref{AKS},\ref{KKKS},\ref{Steane}]).

Let $q$ be a prime power. We use $[[n,k,d]]_q$ to denote a $q$-ary quantum code of length $n$, dimension $q^k$ and minimum distance $d$. Similar to the
classical counterparts, quantum codes have to satisfy the quantum Singleton bound: $k\leq n-2d+2$. The quantum code attaching this bound is called
quantum maximum-distance-separable(MDS) code.

In the past few decades, quantum MDS codes have been extensively studied. The construction of $q$-ary quantum MDS codes with length $n\leq q+1$ has been
investigated from classical Euclidean orthogonal codes (see [\ref{GBR},\ref{RGB}]). On the other hand, some quantum MDS codes with length
$n\geq q+1$ have been investigated, most of which have minimum distances less than $\frac{q}{2}+1$ (see [\ref{JLLX}]). So it is a challenging and valuable
task to construct quantum MDS codes with minimal distances larger than $\frac{q}{2}+1$. Recently, researchers have constructed some of such quantum MDS
codes utilizing constacyclic codes, negacyclic codes and generalized Reed-Solomon codes
(see [\ref{CLZ},\ref{FF1},\ref{FF2},\ref{Guardia}-\ref{JX},\ref{KZ},\ref{KZL},\ref{LX},\ref{LXW},\ref{SYC},\ref{SYZ},\ref{WZ}-\ref{ZG}]).
However, $q$-ary quantum MDS codes with minimal distances bigger than $\frac{q}{2}+1$ are far from complete.

There are dozens of papers on the construction of $[[n,n-2d,d+1]]_q$ quantum MDS codes with relatively large minimum distances. Most of the known
$[[n,n-2d,d+1]]_q$ quantum MDS codes with minimum distances larger than $\frac{q}{2}+1$ have lengths $n\equiv0,1\,(\mathrm{mod}\,q+1)$
(see [\ref{CLZ},\ref{FF1},\ref{GBR},\ref{HXC},\ref{JLLX},\ref{KZ},\ref{KZL},\ref{SYC},\ref{SYZ},\ref{ZG}]) or $n\equiv0,1\,(\mathrm{mod}\,q-1)$
(see [\ref{FF1},\ref{GBR},\ref{HXC}-\ref{JX},\ref{KZ},\ref{SYC},\ref{SYZ},\ref{WZ},\ref{ZG}]), except for the following cases.

(i). $n=q^2-l$ and $d\leq q-l-1$ for $0\leq l\leq q-2 $ (see [\ref{LXW}]).

(ii). $n=mq-l$ and $d\leq m-l$ for $0\leq l<m$ and $1<m<q$ (see [\ref{LXW}] and also [\ref{FF2}] for $l=0$).

(iii). $n=t(q+1)+2$ and $1\leq d \leq t+1$ for $1\leq t\leq q-1$ and $(p,t,d)\neq (2,q-1,q)$ (see [\ref{FF2}] and also [\ref{LXW}] for $t=q-1$).

In this paper, we construct several new classes of quantum MDS codes whose minimum distances can be larger than $\frac{q}{2}+1$ via
generalized Reed-Solomon codes and Hermitian construction. Their lengths are different from the above three cases and also in most cases, are not of
the form $n\equiv0,1\,(\mathrm{mod}\,q\pm1)$. More precisely, the parameters of $[[n,n-2d,d+1]]_q$ quantum MDS codes are as follows:

(i). $n=1+lh+mr-\frac{q^2-1}{st}\cdot hr$ and $1\leq d\leq min\{\frac{s+h}{2}\cdot\frac{q+1}{s}-1,\frac{q+1}{2}+\frac{q-1}{t}-1\}$,
for odd $s\mid q+1$, even $t\mid q-1$, $t\geq 2$, $l=\frac{q^2-1}{s}$, $m=\frac{q^2-1}{t}$, odd $h\leq s-1$, $r\leq t$ and $q-1>\frac{q^2-1}{st}\cdot hr$
(see Theorem \ref{Thm});

(ii). $n=lh+mr-\frac{q^2-1}{st}\cdot hr$ and $1\leq d\leq min\{\lfloor\frac{s+h}{2}\rfloor\cdot\frac{q+1}{s}-2,\frac{q+1}{2}+\frac{q-1}{t}-1\}$,
for odd $s\mid q+1$, even $t\mid q-1$, $t\geq 2$, $l=\frac{q^2-1}{s}$, $m=\frac{q^2-1}{t}$, $h\leq s-1$, $r\leq t$ and
$q-1>\frac{q^2-1}{st}\cdot hr$ (see Theorem \ref{ThmD});

(iii). $n=lh+mr$ and $1\leq d\leq min\{\lfloor\frac{s+h}{2}\rfloor\cdot\frac{q+1}{s}-2,\frac{q+1}{2}+\frac{q-1}{t}-1\}$, for even $s\mid q+1$,
even $t\mid q-1$, $t\geq 2$, $l=\frac{q^2-1}{s}$, $m=\frac{q^2-1}{t}$, $h\leq \frac{s}{2}$ and $r\leq \frac{t}{2}$ (see Theorem \ref{ThmF}).

This paper is organized as follows. In Section 2, we will introduce some basic knowledge and useful results on Hermitian self-orthogonality and
generalized Reed-Solomon codes, which will be utilized in the proof of main results. In Sections 3-5, we will present our main
results on the constructions of quantum MDS codes. In Section 6, we will make a conclusion.

\section{Preliminaries}

 \quad\; In this section, we introduce some basic notations and useful results on Hermitian self-orthogonality and generalized Reed-Solomon codes
(or GRS codes for short).

Let $\mathbb{F}_{q^2}$ be the finite field with $q^2$ elements and $\F_{q^2}^*=\F_{q^2}\backslash\{0\}$, where $q$ is a prime power. Obviously,
$\mathbb{F}_{q}$ is a subfield of $\mathbb{F}_{q^{2}}$ with $q$ elements and denote by $\mathbb{F}_{q}^{*}=\mathbb{F}_{q}\backslash\{0\}$.
For any two vectors $\overrightarrow{x}=(x_1,\ldots,x_n)$ and $\overrightarrow{y}=(y_1,\ldots,y_n)\in\mathbb{F}_{q^2}$,
the Euclidean and Hermitian inner products are defined as $$\langle\overrightarrow{x},\overrightarrow{y}\rangle=\sum\limits_{i=1}^nx_iy_i$$
and $$\langle\overrightarrow{x},\overrightarrow{y}\rangle_H=\sum\limits_{i=1}^nx_iy_i^q$$ respectively.

For a linear code $C$ of length $n$ over $\mathbb{F}_{q^2}$, the Euclidean dual code of $C$ is defined as $$C^\perp:=\{\overrightarrow{x}\in\mathbb{F}_{q^2}^n:\langle\overrightarrow{x},\overrightarrow{y}\rangle=0, \,\,\text{for\,\,all}\,\,\overrightarrow{y}\in C\},$$ and the Hermitian dual code of $C$ is defined as
$$C^{\perp_H}:=\{\overrightarrow{x}\in\mathbb{F}_{q^2}^n:\langle\overrightarrow{x},\overrightarrow{y}\rangle_H=0, \,\,\text{for\,\,all}\,\,\overrightarrow{y}\in C\}.$$
If $C\subseteq C^{\perp_H}$, the code $C$ is called Hermitian self-orthogonal.

In 2001, Ashikhmin and Knill [\ref{AK}] proposed the Hermitian Construction of quantum codes, which is a very important technique for constructing
quantum codes from classical codes.

\begin{theorem}([\ref{AK}, Corollary 1])\label{thmA}
A $q$-ary quantum $[[n,n-2d,d+1]]_q$ MDS code exists provided that an $[n,d,n-d+1]_{q^2}$ MDS Hermitian self-orthogonal code exists.
\end{theorem}

Choose two vectors $\overrightarrow{v}=(v_{1},v_{2},\ldots,v_{n})$ and $\overrightarrow{a}=(a_{1},a_{2},\ldots,a_{n})$,
where $v_{i}\in\mathbb{F}_{q^2}^{*}$ ($v_{i}$ may not be distinct) and $a_{i}$ are distinct elements
in $\mathbb{F}_{q^2}$. For an integer $d$ with $1\leq d\leq n$, the GRS code of length $n$ associated with $\overrightarrow{v}$
and $\overrightarrow{a}$ is defined as follows:
\begin{equation}\label{def GRS}
\mathbf{GRS}_{d}(\overrightarrow{a},\overrightarrow{v})=\{(v_{1}f(a_{1}),\ldots,v_{n}f(a_{n})):f(x)\in\mathbb{F}_{q^2}[x],\mathrm{deg}(f(x))\leq d-1\}.
\end{equation}
The generator matrix of the code $\mathbf{GRS}_{d}(\overrightarrow{a},\overrightarrow{v})$ is
\begin{equation}
G_{d}(\overrightarrow{a},\overrightarrow{v})=
\left(
  \begin{array}{cccc}
    v_1 & v_2 & \cdots & v_n \\
   v_1a_1 & v_2a_{2} & \cdots & v_na_{n} \\
    \vdots & \vdots & \ddots & \vdots \\
    v_1a_1^{d-1} & v_2a_2^{d-1} & \cdots & v_na_n^{d-1} \\
\end{array}
\right).
\end{equation}
It is well known that the code $\mathbf{GRS}_{d}(\overrightarrow{a},\overrightarrow{v})$ is a $q$-ary $[n,d,n-d+1]$ MDS code [\ref{MS}, Chapter 11].
The following theorem will be useful and it has been shown in [\ref{Rain},\ref{ZG}].

\begin{theorem}([\ref{Rain},\ref{ZG}])\label{y2}
The two vectors $\overrightarrow{a}=(a_1,\ldots,a_n)$ and $\overrightarrow{v}=(v_1,\ldots,v_n)$ are defined above. Then
$\mathbf{GRS}_{d}(\overrightarrow{a},\overrightarrow{v})$ is Hermitian self-orthogonal if and only if
$\langle\overrightarrow{a}^{qi+j},\overrightarrow{v}^{q+1}\rangle=0$, for all $0\leq i,j\leq d-1$.
\end{theorem}

If there are no specific statements, the following notations are fixed throughout this paper.

$\bullet$ Let $s\mid q+1$ and $t\mid q-1$ with $t$ even.

$\bullet$ Let $l=\frac{q^{2}-1}{s}$ and $m=\frac{q^{2}-1}{t}$.

$\bullet$ Let $g$ be a primitive element of $\mathbb{F}_{q^2}$, $\delta=g^s$ and $\theta=g^t$.

\begin{lemma}\label{y1}
Suppose $\gcd(s,t)=1$. For any $\alpha,\beta\in\mathbb{Z}_{q^2-1}$, the number of $(i,j)$ of
the equation $\alpha+si\equiv \beta+tj\,(\mathrm{mod}\,q^{2}-1)$ satisfying $0\leq i<\frac{q^2-1}{s}$ and $0\leq j<\frac{q^2-1}{t}$ is $\frac{q^2-1}{st}$.
\end{lemma}

\begin{proof}
Let $\beta-\alpha=\gamma$. From $\alpha+si\equiv \beta+tj\,(\mathrm{mod}\,q^{2}-1)$, we have $si-tj\equiv \gamma\,(\mathrm{mod}\,q^{2}-1)$. When
$0\leq i<\frac{q^2-1}{s}$ and $0\leq j<\frac{q^2-1}{t}$, $si-tj\,\,\mathrm{mod}\,\,q^2-1$ runs $\frac{q^2-1}{st}$ times through
every element of $\mathbb{Z}_{q^2-1}$.

Indeed, for any $\gamma\in\mathbb{Z}_{q^2-1}$, we have $si-tj\equiv \gamma\,(\mathrm{mod}\,q^{2}-1)\Leftrightarrow s\mid tj+\gamma\Leftrightarrow tj\equiv-\gamma\,(\mathrm{mod}\,s)$.
Since $\gcd(s,t)=1$, then $j\,\,\mathrm{mod}\,\,s$ is unique. So when $0\leq j<\frac{q^2-1}{t}$, the number of $j$ satisfying the equation is
$\frac{q^2-1}{st}$. The values of $\gamma$ and $i$ will be determined after fixing $j$. So the number of $(i,j)$ of the equation
$\alpha+si\equiv \beta+tj\,(\mathrm{mod}\,q^{2}-1)$ is $\frac{q^2-1}{st}$ satisfying $0\leq i<\frac{q^2-1}{s}$ and $0\leq j<\frac{q^2-1}{t}$
is $\frac{q^2-1}{st}$.
\end{proof}

The following two lemmas have been shown in [\ref{FF1}] and [\ref{HXC}]. In order to make the paper self completeness, we will give proofs.

\begin{lemma}([\ref{FF1}, Lemmas 5 and 11])\label{lemC}
Assume that $h\leq s-1$.

(i). For any $0\leq i,j\leq \lfloor\frac{s+h}{2}\rfloor\cdot\frac{q+1}{s}-3$, $l\mid (qi+j+q+1)$ if and only if $qi+j+q+1=\mu\cdot l$,
with $\lceil\frac{s-h}{2}\rceil+1\leq\mu\leq\lfloor\frac{s+h}{2}\rfloor-1$.

(ii). For any $0\leq i,j\leq \lfloor\frac{s+h}{2}\rfloor\cdot\frac{q+1}{s}-2$ with $(i,j)\neq(0,0)$, $l\mid (qi+j)$
if and only if $qi+j=\mu\cdot l$, with $\lceil\frac{s-h}{2}\rceil+1\leq\mu\leq \lfloor\frac{s+h}{2}\rfloor-1$.
\end{lemma}

\begin{proof}
(i). When $s\equiv h\,(\mathrm{mod}\,2)$, it implies $\lfloor\frac{s+h}{2}\rfloor=\frac{s+h}{2}$ and $\lceil\frac{s-h}{2}\rceil=\frac{s-h}{2}$.
Since $0\leq i,j\leq \frac{s+h}{2}\cdot\frac{q+1}{s}-3<q-2$, then $0<qi+j+q+1< q^2-1$, that is $0<\mu<s$.
From $qi+j+q+1=q\left(\frac{\mu\cdot(q+1)}{s}-1\right)+\left(q-\frac{\mu\cdot(q+1)}{s}\right)$, it follows that
$$i=\frac{\mu\cdot(q+1)}{s}-2,j=q-\frac{\mu\cdot(q+1)}{s}-1.$$ By $i<\frac{s+h}{2}\cdot\frac{q+1}{s}-2$ and $j<\frac{s+h}{2}\cdot\frac{q+1}{s}-2$,
it implies $\frac{s-h}{2}<\mu<\frac{s+h}{2}$. So $l\mid (qi+j)$ if and only if $qi+j=\mu\cdot l$, with $\frac{s-h}{2}+1\leq\mu\leq \frac{s+h}{2}-1$.

When $s\not\equiv h\,(\mathrm{mod}\,2)$, it implies $\lfloor\frac{s+h}{2}\rfloor=\frac{s+h-1}{2}$ and $\lceil\frac{s-h}{2}\rceil=\frac{s-h+1}{2}$.
Then the proof can be completed by proceeding as the situation that $s\equiv h\,(\mathrm{mod}\,2)$.

(ii). In a similar way, we can complete the proof. So we omit the details.
\end{proof}

\begin{lemma}([\ref{HXC}, Lemma 3.1])\label{yC}
The identity $\sum\limits_{\nu=0}^{m-1}\theta^{\nu\left(qi+j+\frac{q+1}{2}\right)}=0$ holds for all $0\leq i,j\leq \frac{q+1}{2}+\frac{q-1}{t}-2$,
with even $t\geq2$.
\end{lemma}

\begin{proof}
It is easy to check that the identity holds if and only if $m\nmid qi+j+\frac{q+1}{2}$. On the contrary, assume that $m\mid qi+j+\frac{q+1}{2}$. Let
\begin{equation}\label{rem}
qi+j+\frac{q+1}{2}=\mu\cdot m=q\cdot\frac{\mu(q-1)}{t}+\frac{\mu(q-1)}{t}
\end{equation}
with $\mu\in\mathbb{Z}$. By $t\geq2$, we have $qi+j+\frac{q+1}{2}<q^2-1$, which implies $0<\mu<t$.

\begin{itemize}
\item If $j+\frac{q+1}{2}\leq q-1$, comparing remainder and quotient of module $q$ on both sides of (\ref{rem}), we can deduce $i=j+\frac{q+1}{2}=\mu\cdot \frac{q-1}{t}$.
Since $t$ is even, then $\frac{q-1}{t}\mid\frac{q-1}{2}$. From $\frac{q-1}{t}\mid  j+1+\frac{q-1}{2}$, we can deduce that $\frac{q-1}{t}\mid j+1$.
Since $j+1\geq1$, then $j+1\geq\frac{q-1}{t}$. So $i=j+\frac{q+1}{2}\geq\frac{q+1}{2}+\frac{q-1}{t}-1$, which is a contradiction.

\item When $j+\frac{q+1}{2}\geq q$, it takes $qi+j+\frac{q+1}{2}=q(i+1)+\left(j-\frac{q-1}{2}\right)=q\cdot\frac{\mu(q-1)}{t}+\frac{\mu(q-1)}{t}$. In a similar way,
$j-\frac{q-1}{2}=i+1=\mu\cdot\frac{q-1}{t}$ which implies $\frac{q-1}{t}\mid i+1$. Since $i+1\geq1$, then $i+1\geq\frac{q-1}{t}$.
Therefore, $j=i+1+\frac{q-1}{2}\geq\frac{q+1}{2}+\frac{q-1}{t}-1$, which is a contradiction.
\end{itemize}
As a result, $m\nmid qi+j+\frac{q+1}{2}$ which yields $\sum\limits_{\nu=0}^{m-1}\theta^{\nu(qi+j+\frac{q+1}{2})}=0$ for all
$0\leq i,j\leq \frac{q+1}{2}+\frac{q-1}{t}-2$.
\end{proof}

\section{Quantum MDS Codes of Length $n=1+lh+mr-\frac{q^2-1}{st}\cdot hr$}

 \quad\; In this section, we assume that $\mathbf{s\,\, is\,\, odd}$, $\mathbf{h\leq s-1}$ with $\mathbf{h\,\,odd}$ and $\mathbf{r\leq t}$.
Quantum MDS codes of length $n=1+lh+mr-\frac{q^2-1}{st}\cdot hr$ will be constructed. The construction is based on [\ref{FF1}] and [\ref{HXC}].
Firstly, we choose elements in $\mathbb{F}_{q^2}^*/\langle\delta\rangle$ as the first part of coordinates in the vector $\overrightarrow{a}$.
Secondly, we choose elements from cosets of $\mathbb{F}_{q^2}^*/\langle\theta\rangle$ as the second part of coordinates in $\overrightarrow{a}$.
Finally, we consider the duplicating elements between these two parts. We construct the vector $\overrightarrow{v}$ in a similar way.
Then we can construct quantum MDS codes of length $n=1+lh+mr-\frac{q^2-1}{st}\cdot hr$, whose minimum distances can be bigger than $\frac{q}{2}+1$.

The next lemma has been shown in [\ref{FF1}]. We give a new proof by Cramer's Rule, which is shorter than [\ref{FF1}].

\begin{lemma}([\ref{FF1}, Lemma 7])\label{lemD}
For $\frac{s-h}{2}+1\leq \mu\leq\frac{s+h}{2}-1$, there exists a solution in $(\mathbb{F}_{q}^*)^{h}$ of the following system of equations
\begin{equation}\label{equB}
\begin{cases}
u_{0}+u_{1}+\cdots+u_{h-1}=1\\
\sum\limits_{k=0}^{h-1}g^{k\mu l}u_{k}=0&
\end{cases}
\end{equation}
\end{lemma}

\begin{proof}
Denote by $\xi=g^l$ and $c=\frac{s-h}{2}+1$. For any $0\leq\nu\neq\nu'\leq h-2<s-2$, the elements $\xi^{c+\nu}$, $\xi^{c+\nu'}$ and $1$ are
distinct. The system of equations (\ref{equB}) can be expressed in the matrix form
\begin{equation}\label{equC}
\begin{aligned}
A\overrightarrow{u}^{T}=(1,0,\cdots,0)^T,
\end{aligned}
\end{equation}
where \begin{equation*}
A=
\left(
  \begin{array}{cccc}
    1 & 1 & \cdots & 1 \\
   1 & \xi^{c} & \cdots & \xi^{(h-1)c} \\
     \vdots & \vdots & \ddots & \vdots \\
   1 & \xi^{c+h-2}  & \cdots & \xi^{(h-1)(c+h-2)} \\
\end{array}
\right)_{h\times h}
\end{equation*}
and $$\overrightarrow{u}=(u_{0},u_{1},\ldots,u_{h-1}).$$ We will show that $u_k\in\F_q^*$ for any $0\leq k\leq h-1$.

It is obvious that $\mathrm{det}(A)\neq0$. By Cramer's Rule, $$u_k=\frac{(-1)^{k}\cdot \mathrm{det}(D_{k})}{\mathrm{det}(A)},$$ where
\begin{equation}
D_{k}=
\left(
  \begin{array}{ccccccc}
   1 & \xi^{c} & \cdots & \xi^{(k-1)c} & \xi^{(k+1)c} & \cdots & \xi^{(h-1)c} \\
   1 & \xi^{c+1} & \cdots  & \xi^{(k-1)(c+1)} & \xi^{(k+1)(c+1)} & \cdots & \xi^{(h-1)(c+1)} \\
   \vdots & \vdots & \ddots & \vdots & \vdots & \ddots & \vdots \\
   1 & \xi^{c+h-2} & \cdots & \xi^{(k-1)(c+h-2)} & \xi^{(k+1)(c+h-2)} & \cdots & \xi^{(h-1)(c+h-2)} \\
\end{array}
\right)
\end{equation}
is an $(h-1)\times(h-1)$ matrix obtained from $A$ by deleting $1$-st row and $(k+1)$-th column with $0\leq k\leq h-1$. It is easy to see
$\mathrm{det}(D_{k})$ is equal to non-zero constant times of a Vandermonde determinant. So $\mathrm{det}(D_{k})\neq0$, which implies $u_k\neq0$.

It remains to show $u_{k}\in\mathbb{F}_q$, for any $0\leq k\leq h-1$. Since $s\mid q+1$ and $\xi^s=1$, then
$$\xi^{k(c+\nu)q}=\xi^{-k\left(\frac{s-h}{2}+1+\nu\right)}=\xi^{k\left(\frac{s+h}{2}-1-\nu\right)}=\xi^{k(c+h-2-\nu)},$$ for any $0\leq k\leq h-1$ and
$0\leq \nu\leq h-2$. So $(\mathrm{det}(A))^q=(-1)^{\frac{h-1}{2}}\cdot\mathrm{det}(A)$ and
$\mathrm{det}(D_{k})^{q}=(-1)^{\frac{h-1}{2}}\cdot \mathrm{det}(D_{k})$. It follows that
$u_k^q=\frac{(-1)^{qk}\cdot\mathrm{det}(D_{k})^q}{(\mathrm{det}(A))^q}=\frac{(-1)^{k}\cdot \mathrm{det}(D_{k})}{\mathrm{det}(A)}=u_k$,
which implies $u_k\in\mathbb{F}_{q}^*$ with $0\leq k\leq h-1$. This completes the proof.
\end{proof}

Now we let $\overrightarrow{u}=(u_0,u_1,\ldots,u_{h-1})$ satisfy the system of equations (\ref{equB}). Choose
$$\overrightarrow{a}_1=(0,1,\delta,\ldots,\delta^{l-1},g,g\delta,\ldots,g\delta^{l-1},\ldots, g^{h-1},g^{h-1}\delta,\ldots,g^{h-1}\delta^{l-1})$$ and
$$\overrightarrow{v}_1=(e,\underbrace{v_{0},\ldots,v_{0}}_{l\,\, \text{times}},\ldots,\underbrace{v_{h-1},\ldots,v_{h-1}}_{l\,\, \text{times}}),$$
where $v_{k}^{q+1}=u_k$($0\leq k\leq h-1$) and $e^{q+1}=-l$. Then we have the following lemma, which has been shown in [\ref{FF1}]. We give proof in order
to make the paper self completeness.

\begin{lemma}([\ref{FF1}, Theorem 3])\label{thmB}
The identity $$\langle\overrightarrow{a}_1^{qi+j},\overrightarrow{v}_1^{q+1}\rangle=0$$ holds for all $0\leq i,j\leq \frac{s+h}{2}\cdot\frac{q+1}{s}-2$.
\end{lemma}

\begin{proof}
When $(i,j)=(0,0)$,
$$\langle\overrightarrow{a}_1^{0},\overrightarrow{v}_1^{q+1}\rangle=e^{q+1}+l(v_{0}^{q+1}+\cdots+v_{h-1}^{q+1})=-l+l(u_0+\cdots+ u_{h-1})=0.$$
When $(i,j)\neq(0,0)$, since $\delta$ is of order $l$, then
$$\langle\overrightarrow{a}_1^{qi+j},\overrightarrow{v}_1^{q+1}\rangle=\sum\limits_{k=0}^{h-1}g^{k(qi+j)}v_{k}^{q+1} \sum\limits_{\nu=0}^{l-1}\delta^{\nu(qi+j)}=
\begin{cases}
0,& \text{$l\nmid qi+j$,}\\
l\cdot\sum\limits_{k=0}^{h-1}g^{k(qi+j)}v_{k}^{q+1},& \text{$l\mid qi+j$.}
\end{cases}$$
We consider the case $l\mid qi+j$. According to Lemma \ref{lemC} (ii) and Lemma \ref{lemD},
$$\langle\overrightarrow{a}_1^{qi+j},\overrightarrow{v}_1^{q+1}\rangle=\langle\overrightarrow{a}_1^{\mu l},\overrightarrow{v}_1^{q+1}\rangle
=l\cdot\sum\limits_{k=0}^{h-1}g^{k\mu l}v_{k}^{q+1}=l\cdot\sum\limits_{k=0}^{h-1}g^{k\mu l}u_{k}=0.$$ Therefore, the result holds.
\end{proof}

For the second part of $\overrightarrow{a}$ and $\overrightarrow{v}$,
we choose $$\overrightarrow{a}_2=(1,\theta,\ldots,\theta^{m-1},g,g\theta,\ldots,g\theta^{m-1}, \ldots,g^{r-1},g^{r-1}\theta,\ldots,g^{r-1}\theta^{m-1})$$
and $$\overrightarrow{v}_2=(1,g^{\frac{t}{2}},\ldots,g^{(m-1)\cdot\frac{t}{2}},1,g^{\frac{t}{2}}, \ldots,g^{(m-1)\cdot\frac{t}{2}},\ldots,1,g^{\frac{t}{2}},\ldots,g^{(m-1)\cdot\frac{t}{2}}).$$
Then the following lemma can be obtained.

\begin{lemma}\label{thmC}
The identity $$\langle\overrightarrow{a}_2^{qi+j},\overrightarrow{v}_2^{q+1}\rangle=0$$ holds for all $0\leq i,j\leq\frac{q+1}{2}+\frac{q-1}{t}-2$.
\end{lemma}

\begin{proof}
By Lemma \ref{yC}, we can calculate directly,
\begin{equation}
\begin{aligned}
\langle\overrightarrow{a}_2^{qi+j},\overrightarrow{v}_2^{q+1}\rangle&=\sum\limits_{k=0}^{r-1}\sum\limits_{\nu=0}^{m-1}
(g^k\theta^{\nu})^{qi+j}\cdot\theta^{\nu\cdot\frac{q+1}{2}}&\\
&=\sum\limits_{k=0}^{r-1}g^{k(qi+j)}\sum\limits_{\nu=0}^{m-1}\theta^{\nu(qi+j+\frac{q+1}{2})}&\\
&=0.&
\end{aligned}
\end{equation}
\end{proof}

Now, we give our first construction.

\begin{theorem}\label{Thm}
Let $n=1+lh+mr-\frac{q^2-1}{st}\cdot hr$, where odd $s\mid q+1$, even $t\mid q-1$, $t\geq 2$, $l=\frac{q^2-1}{s}$, $m=\frac{q^2-1}{t}$, odd $h\leq s-1$
and $r\leq t$. If $q-1>\frac{q^2-1}{st}\cdot hr$, then for any $1\leq d\leq min\{\frac{s+h}{2}\cdot\frac{q+1}{s}-1,\frac{q+1}{2}+\frac{q-1}{t}-1\}$,
there exists an $[[n,n-2d,d+1]]_q$ quantum MDS code.
\end{theorem}
\begin{proof}
Denote by \[A=\{g^\alpha\delta^i|0\leq\alpha\leq h-1,0\leq i\leq l-1\}\] and \[B=\{g^\beta\theta^j|0\leq\beta\leq r-1,0\leq j\leq m-1\}.\]
From Lemma \ref{y1}, we know $|A\cap B|=\frac{q^2-1}{st}\cdot hr$. Let $A_1=A-B$ and $B_1=B-A$. Define
\[
\begin{array}{lcr}
f_1:A\cup\{0\}\rightarrow \mathbb{F}_q^*,\,\,f_1(g^\alpha\delta^i)=v_\alpha^{q+1}\,\, \text{and}\,\, f_1(0)=-l,\\[2mm]
f_2:B\rightarrow \mathbb{F}_q^*,\,\,f_2(g^\beta\theta^j)=\theta^{j\cdot\frac{q+1}{2}}.\\[2mm]
\end{array}
\]
Let $$\overrightarrow{a}=(0,\overrightarrow{a}_{A_1},\overrightarrow{a}_{B_1},\overrightarrow{a}_{A\cap B}),$$
where $\overrightarrow{a}_{S}=(a_1,\ldots,a_k)$ for $S=\{a_1,\ldots,a_k\}$ and
$$\overrightarrow{v}^{q+1}=(-l, f_1(\overrightarrow{a}_{A_1}),\lambda f_2(\overrightarrow{a}_{B_1}),f_1(\overrightarrow{a}_{A\cap B})+\lambda f_2(\overrightarrow{a}_{A\cap B})),$$
where $\lambda\in\mathbb{F}_q^*$ and $f_j(\overrightarrow{a}_{S})=(f_j(a_1),\ldots,f_j(a_k))$ with $S=\{a_1,\ldots,a_k\}$ and $j=1,2$.

Indeed, since $q-1>\frac{q^2-1}{st}\cdot hr=|A\cap B|$, then there exists $\lambda\in\mathbb{F}_q^*$ such that all coordinates of
$f_1(\overrightarrow{a}_{A\cap B})+\lambda f_2(\overrightarrow{a}_{A\cap B})$ are nonzero.

According to Lemmas \ref{thmB} and \ref{thmC}, it takes
$$\langle\overrightarrow{a}^{qi+j},\overrightarrow{v}^{q+1}\rangle=\langle\overrightarrow{a}_1^{qi+j},\overrightarrow{v}_1^{q+1}\rangle +\lambda\langle\overrightarrow{a}_2^{qi+j},\overrightarrow{v}_2^{q+1}\rangle=0,$$
for any $0\leq i,j\leq d-1$. As a consequence, by Theorem \ref{y2}, $\mathbf{GRS}_{d}(\overrightarrow{a},\overrightarrow{v})$ is Hermitian self-orthogonal.
Therefore, by Theorem \ref{thmA}, there exists an $[[n,n-2d,d+1]]_q$ quantum MDS code, where
$n=1+lh+mr-\frac{q^2-1}{st}\cdot hr$ and $1\leq d\leq min\{\frac{s+h}{2}\cdot\frac{q+1}{s}-1,\frac{q+1}{2}+\frac{q-1}{t}-1\}$.
\end{proof}

\begin{remark}\label{R1}
We try to choose $s,h,t$ such that $\frac{s+h}{2}\cdot\frac{q+1}{s}-1\approx \frac{q+1}{2}+\frac{q-1}{t}-1$. For large $q$,
we take $s\approx\frac{1}{2}\sqrt{2(q+1)}\cdot h$ and $t\approx\sqrt{2(q+1)}$. Then it follows that
$$\frac{s+h}{2}\cdot\frac{q+1}{s}-1\approx\frac{q}{2}+\sqrt{\frac{q}{2}}\,\,\,\, \text{and}\,\,\,\,\frac{q+1}{2}+\frac{q-1}{t}-1\approx\frac{q}{2}+\sqrt{\frac{q}{2}}.$$
This indicates that the minimum distance of the quantum MDS code in Theorem \ref{Thm} can reach $\frac{q}{2}+\sqrt{\frac{q}{2}}$ approximately.
\end{remark}

\begin{example}
Let $q=641$. Choose $s=107$, $t=32$, $h=5$ and $r=1$. In this case, one has $\frac{s+h}{2s}\cdot(q+1)-1=341$ and
$\frac{q+1}{2}+\frac{q-1}{t}-1=340\approx\frac{q}{2}+\sqrt{\frac{q}{2}}=338.4$. The length is $n=1+lh+mr-\frac{q^2-1}{st}\cdot hr=16081$.
There exists $[[16081,15401,341]]_{641}$ quantum MDS code, which has not been covered in any previous work.
\end{example}

\section{Quantum MDS Codes of Length $n=lh+mr-\frac{q^2-1}{st}\cdot hr$}

\quad\; In this section, we assume $\mathbf{s\,\, is\,\, odd}$,
$\mathbf{h\leq s-1}$ and $\mathbf{r\leq t}$. Now, we consider the first part of coordinates in vectors $\overrightarrow{a}$ and $\overrightarrow{v}$. Firstly, we give two useful lemmas, that are
Lemmas \ref{lemDD} and \ref{thmBB}, which generalize Lemma 13 and Theorem 5 in [\ref{FF1}], respectively.

\begin{lemma}\label{lemDD}
There exists a solution in $(\mathbb{F}_{q}^*)^{h}$ of the following system of equations
\begin{equation}\label{equBB}
\begin{aligned}
\sum\limits_{k=0}^{h-1}g^{k(\mu l-q-1)}u_{k}=0
\end{aligned}
\end{equation}
for $\lceil\frac{s-h}{2}\rceil+1\leq\mu\leq\lfloor\frac{s+h}{2}\rfloor-1$.
\end{lemma}

\begin{proof}
Let $\xi=g^l$, $\eta=g^{-q-1}\in\mathbb{F}_q^*$ and $c=\lceil\frac{s-h}{2}\rceil+1$. It is clear that
$\xi^{c+\nu}\neq\xi^{c+\nu'}$ for any $0\leq\nu\neq\nu'\leq h-2<s-2$.  We discuss in two cases.

{\bf Case 1}: $h$ is odd. In this case, $\lceil\frac{s-h}{2}\rceil=\frac{s-h}{2}$ and $\lfloor\frac{s+h}{2}\rfloor=\frac{s+h}{2}$.
The system of equations (\ref{equBB}) can be expressed in the matrix form
\begin{equation}\label{equC}
\begin{aligned}
A\overrightarrow{u}^{T}=(0,0,\ldots,0)^T,
\end{aligned}
\end{equation}
where \begin{equation*}
A=
\left(
  \begin{array}{ccccc}
   1 & \xi^{c}\eta & \xi^{2c}\eta^2 & \cdots & \xi^{(h-1)c}\eta^{h-1} \\
   1 & \xi^{c+1}\eta & \xi^{2(c+1)}\eta^2 & \cdots & \xi^{(h-1)(c+1)}\eta^{h-1} \\
    \vdots &  \vdots & \vdots & \ddots & \vdots \\
   1 & \xi^{c+h-2}\eta & \xi^{2(c+h-2)}\eta^2 & \cdots & \xi^{(h-1)(c+h-2)}\eta^{h-1} \\
\end{array}
\right)
\end{equation*}
is an $(h-1)\times h$ matrix over $\mathbb{F}_{q^2}$ and $$\overrightarrow{u}=(u_{0},u_{1},\ldots,u_{h-1}).$$
It is obvious that $\mathrm{rank}(A)=h-1$. We will show that $u_k\in\F_q^*$ for any $0\leq k\leq h-1$.

Let \begin{equation*}
A'=
\left(
  \begin{array}{ccccc}
   1 & 1 & 1 & \cdots & 1 \\
   1 & \xi^{c}\eta & \xi^{2c}\eta^2 & \cdots & \xi^{(h-1)c}\eta^{h-1} \\
   1 & \xi^{c+1}\eta & \xi^{2(c+1)}\eta^2 & \cdots & \xi^{(h-1)(c+1)}\eta^{h-1} \\
    \vdots &  \vdots & \vdots & \ddots & \vdots \\
   1 & \xi^{c+h-2}\eta & \xi^{2(c+h-2)}\eta^2 & \cdots & \xi^{(h-1)(c+h-2)}\eta^{h-1} \\
\end{array}
\right).
\end{equation*}
We consider the equations
\begin{equation}\label{equC'}
\begin{aligned}
A'\overrightarrow{u}^{T}=(1,0,0,\ldots,0)^T.
\end{aligned}
\end{equation}
It is easy to check that $A'$ is row equivalent to $A'^{(q)}$ and $\det(A')\neq 0$. Similarly as the proof of Lemma \ref{lemD},
we obtain (\ref{equC'}) has a solution $\overrightarrow{u}=(u_{0},u_{1},\ldots,u_{h-1})\in(\mathbb{F}_{q}^*)^{h}$.
Since the solution of (\ref{equC'}) is also the solution of (\ref{equC}), the result has been proved.

{\bf Case 2}: $h$ is even. In this case, $\lceil\frac{s-h}{2}\rceil=\frac{s-h+1}{2}$ and $\lfloor\frac{s+h}{2}\rfloor=\frac{s+h-1}{2}$.
The system of equations (\ref{equBB}) can be expressed in the matrix form
\begin{equation}\label{equCC}
\begin{aligned}
A\overrightarrow{u}^{T}=(0,0,\ldots,0)^T,
\end{aligned}
\end{equation}
where \begin{equation*}
A=
\left(
  \begin{array}{ccccc}
   1 & \xi^{c}\eta & \xi^{2c}\eta^2 & \cdots & \xi^{(h-1)c}\eta^{h-1} \\
   1 & \xi^{c+1}\eta & \xi^{2(c+1)}\eta^2 & \cdots & \xi^{(h-1)(c+1)}\eta^{h-1} \\
    \vdots &  \vdots & \vdots & \ddots & \vdots \\
   1 & \xi^{c+h-3}\eta & \xi^{2(c+h-3)}\eta^2 & \cdots & \xi^{(h-1)(c+h-3)}\eta^{h-1} \\
\end{array}
\right)
\end{equation*}
is an $(h-2)\times h$ matrix over $\mathbb{F}_{q^2}$. By $s\mid q+1$ and $\xi^s=1$, it takes
$$\left(\xi^{k(c+\nu)}\eta^k\right)^q=\xi^{-k\left(\frac{s-h+1}{2}+1+\nu\right)}\eta^k
=\xi^{k\left(\frac{s+h-1}{2}-1-\nu\right)}\eta^k=\xi^{k(c+h-3-\nu)}\eta^k,$$
for any $0\leq k\leq h-1$ and $0\leq \nu\leq h-3$. Therefore, $A$ and $A^{(q)}$ are row equivalent.
By deleting the first (resp. the last) column of $A$ and we obtain an $(h-2)\times(h-1)$ matrix denote by $A_0$ (resp. $A_{h-1}$). Then
$A_0$ (resp. $A_{h-1}$) is row equivalent to $A_0^{(q)}$ (resp. $A_{h-1}^{(q)}$). Obviously, $\mathrm{rank}(A_0)=\mathrm{rank}(A_{h-1})=h-2$.
Similarly as the proof of Case 1, we can deduce that the following equations
$$A_0\overrightarrow{x}^T=(0,\ldots,0)^T,A_{h-1}\overrightarrow{y}^T=(0,\ldots,0)^T$$
have two solutions $\overrightarrow{x}=(x_1,x_2,\ldots,x_{h-1}),\overrightarrow{y}=(y_0,y_1,\ldots,y_{h-2})\in(\mathbb{F}_q^*)^{h-1}$.
From $h<q+1$, there exists $\lambda\in\mathbb{F}_q^*\setminus\{\frac{x_1}{y_1},\ldots,\frac{x_{h-2}}{y_{h-2}}\}$ such that
$\overrightarrow{u}=(0,\overrightarrow{x})-\lambda(\overrightarrow{y},0)\in(\mathbb{F}_q^*)^h$. Then it implies
$$A\overrightarrow{u}^T=\left( \begin{array}{c}
0 \\
A_0\overrightarrow{x}^T \\
\end{array} \right)
-\lambda\left( \begin{array}{c} A_{h-1}\overrightarrow{y}^T \\
 0 \\
 \end{array} \right)=(0,0,\ldots,0)^T.$$
Therefore, the result has been proved.
\end{proof}

We choose
$$\overrightarrow{a}_1=(1,\delta,\ldots,\delta^{l-1},g,g\delta,\ldots,g\delta^{l-1},\ldots, g^{h-1},g^{h-1}\delta,\ldots,g^{h-1}\delta^{l-1})$$ and
$$\overrightarrow{v}_1=(v_0,v_0\delta,\ldots,v_0\delta^{l-1},v_1,v_1\delta,\ldots,v_1\delta^{l-1},\ldots, v_{h-1},v_{h-1}\delta,\ldots,v_{h-1}\delta^{l-1}),$$
where $v_{k}^{q+1}=u_k$($0\leq k\leq h-1$) and $\overrightarrow{u}=(u_0,u_1,\ldots,u_{h-1})$ satisfy (\ref{equBB}).

\begin{lemma}\label{thmBB}
The identity $$\langle\overrightarrow{a}_1^{qi+j},\overrightarrow{v}_1^{q+1}\rangle=0$$ holds for all
$0\leq i,j\leq \lfloor\frac{s+h}{2}\rfloor\cdot\frac{q+1}{s}-3$.
\end{lemma}

\begin{proof}
Similarly as Lemma \ref{thmB}, we only need to consider the case $l\mid qi+j+q+1$. From Lemma \ref{lemC} (i) and Lemma \ref{lemDD}, we deduce that
$$\langle\overrightarrow{a}_1^{qi+j},\overrightarrow{v}_1^{q+1}\rangle=\langle\overrightarrow{a}_1^{\mu l-q-1},\overrightarrow{v}_1^{q+1}\rangle
=l\cdot\sum\limits_{k=0}^{h-1}g^{k(\mu l-q-1)}v_{k}^{q+1}=0.$$ Therefore,
for all $0\leq i,j\leq \lfloor\frac{s+h}{2}\rfloor\cdot\frac{q+1}{s}-3$, $$\langle\overrightarrow{a}_1^{qi+j},\overrightarrow{v}_1^{q+1}\rangle=0.$$
\end{proof}

The vectors $\overrightarrow{a}_2$ and $\overrightarrow{v}_2$ are the same as in Section 3.

\begin{theorem}\label{ThmD}
Let $n=lh+mr-\frac{q^2-1}{st}\cdot hr$, where odd $s\mid q+1$, even $t\mid q-1$, $t\geq 2$, $l=\frac{q^2-1}{s}$, $m=\frac{q^2-1}{t}$, $h\leq s-1$ and
$r\leq t$. Assume that $q-1>\frac{q^2-1}{st}\cdot hr$, then for any
$1\leq d\leq min\{\lfloor\frac{s+h}{2}\rfloor\cdot\frac{q+1}{s}-2,\frac{q+1}{2}+\frac{q-1}{t}-1\}$, there exists an $[[n,n-2d,d+1]]_q$ quantum MDS code.
\end{theorem}

\begin{proof}
Similarly as Theorem \ref{Thm}, we also let $A=\{g^\alpha\delta^i|0\leq\alpha\leq h-1,0\leq i\leq l-1\}$,
$B=\{g^\beta\theta^j|0\leq\beta\leq r-1,0\leq j\leq m-1\}$, $A_1=A-B$ and $B_1=B-A$. Define
\[
\begin{array}{lcr}
f_1:A\rightarrow \mathbb{F}_q^*,\,\,f_1(g^\alpha\delta^i)=(v_\alpha\delta^i)^{q+1},\\[2mm]
f_2:B\rightarrow \mathbb{F}_q^*,\,\,f_2(g^\beta\theta^j)=\theta^{j\cdot\frac{q+1}{2}}.\\[2mm]
\end{array}
\]
Let $$\overrightarrow{a}=(\overrightarrow{a}_{A_1},\overrightarrow{a}_{B_1},\overrightarrow{a}_{A\cap B}),$$
where $\overrightarrow{a}_{S}=(a_1,\ldots,a_k)$ for $S=\{a_1,\ldots,a_k\}$ and
$$\overrightarrow{v}^{q+1}=(f_1(\overrightarrow{a}_{A_1}),\lambda f_2(\overrightarrow{a}_{B_1}),f_1(\overrightarrow{a}_{A\cap B})+\lambda f_2(\overrightarrow{a}_{A\cap B})),$$
where  $\lambda\in\mathbb{F}_q^*$ is chosen such that all the coordinates of $f_1(\overrightarrow{a}_{A\cap B})+\lambda f_2(\overrightarrow{a}_{A\cap B})$ are nonezero  and $f_j(\overrightarrow{a}_{S})=(f_j(a_1),\ldots,f_j(a_k))$ with $S=\{a_1,\ldots,a_k\}$ for $j=1,2$.

According to Lemmas  \ref{thmC} and \ref{thmBB}, similarly as the proof of Theorem \ref{Thm}, $\mathbf{GRS}_{d}(\overrightarrow{a},\overrightarrow{v})$ is
Hermitian self-orthogonal. As a consequence, by Theorem \ref{thmA}, there exists $[[n,n-2d,d+1]]_q$ quantum MDS code, where
$n=lh+mr-\frac{q^2-1}{st}\cdot hr$ with odd $h$ and $1\leq d\leq min\{\lfloor\frac{s+h}{2}\rfloor\cdot\frac{q+1}{s}-2,\frac{q+1}{2}+\frac{q-1}{t}-1\}$.
\end{proof}

\begin{remark}
Similarly as Remark \ref{R1}, the minimum distance can reach $\frac{q}{2}+\sqrt{\frac{q}{2}}$ approximately.
\end{remark}

\section{Quantum MDS Codes of Length $n=lh+mr$}

\quad\; In this section,  $\mathbf{s\,\,is\,\,even}$, $\mathbf{h\leq\frac{s}{2}}$ and $\mathbf{r\leq \frac{t}{2}}$
and quantum MDS codes with length $n=lh+mr$ will be constructed. Similarly as the previous constructions, we also divide the vectors
$\overrightarrow{a}$ and $\overrightarrow{v}$ into two parts. However, in this case, coordinates of these two parts in the vector $\overrightarrow{a}$
have no duplication. Therefore, the quantum MDS codes in this section have larger minimum distances than the codes in previous sections.

The proof of the next result is similar to that of Lemma \ref{lemDD} and we omit the details.

\begin{lemma}\label{lemF2}
The following system of equations
\begin{equation}\label{equF1}
\begin{aligned}
\sum\limits_{k=0}^{h-1}g^{(2k+1)(\mu l-q-1)}u_{k}=0
\end{aligned}
\end{equation}
has a solution denote by $\overrightarrow{u}=(u_{0},u_{1},\ldots,u_{h-1})\in(\mathbb{F}_{q}^*)^{h}$ for all
$\lceil\frac{s-h}{2}\rceil+1\leq\mu\leq\lfloor\frac{s+h}{2}\rfloor-1$.
\end{lemma}

Here we choose
$$\overrightarrow{a}_1=(g,g\delta,\ldots,g\delta^{l-1},g^3,g^3\delta,\ldots,g^3\delta^{l-1},\ldots, g^{2h-1},g^{2h-1}\delta,\ldots,g^{2h-1}\delta^{l-1})$$
and $$\overrightarrow{v}_1=(v_0,v_0\delta,\ldots,v_0\delta^{l-1},v_1,v_1\delta,\ldots,v_1\delta^{l-1},\ldots, v_{h-1},v_{h-1}\delta,\ldots,v_{h-1}\delta^{l-1}),$$
where $v_{k}^{q+1}=u_k$($0\leq k\leq h-1$) and $\overrightarrow{u}=(u_0,u_1,\ldots,u_{h-1})$ is a solution of (\ref{equF1}).

\begin{lemma}\label{thmF3}
The identity $$\langle\overrightarrow{a}_1^{qi+j},\overrightarrow{v}_1^{q+1}\rangle=0$$ holds for all
$0\leq i,j\leq \lfloor\frac{s+h}{2}\rfloor\cdot\frac{q+1}{s}-3$.
\end{lemma}

\begin{proof}
The result follows from Lemmas \ref{lemC} (i) and \ref{lemF2}.
\end{proof}

Now we construct the second part of coordinates in $\overrightarrow{a}$ and $\overrightarrow{v}$.
We choose $$\overrightarrow{a}_2=(1,\theta,\ldots,\theta^{m-1},g^2,g^2\theta,\ldots,g^2\theta^{m-1}, \ldots,g^{2r-2},g^{2r-2}\theta,\ldots,g^{2r-2}\theta^{m-1})$$
and $$\overrightarrow{v}_2=(1,g^{\frac{t}{2}},\ldots,g^{(m-1)\cdot\frac{t}{2}},1,g^{\frac{t}{2}}, \ldots,g^{(m-1)\cdot\frac{t}{2}},\ldots,1,g^{\frac{t}{2}},\ldots,g^{(m-1)\cdot\frac{t}{2}}).$$
Then we have the following lemma.

\begin{lemma}\label{thmF4}
The identity $$\langle\overrightarrow{a}_2^{qi+j},\overrightarrow{v}_2^{q+1}\rangle=0$$ holds for all $0\leq i,j\leq\frac{q+1}{2}+\frac{q-1}{t}-2$.
\end{lemma}

\begin{proof}
By Lemma \ref{yC},
\begin{equation}
\begin{aligned}
\langle\overrightarrow{a}_2^{qi+j},\overrightarrow{v}_2^{q+1}\rangle&=\sum\limits_{k=0}^{r-1}\sum\limits_{\nu=0}^{m-1}
(g^{2k}\theta^{\nu})^{qi+j}\cdot\theta^{\nu\cdot\frac{q+1}{2}}&\\
&=\sum\limits_{k=0}^{r-1}g^{2k(qi+j)}\sum\limits_{\nu=0}^{m-1}\theta^{\nu(qi+j+\frac{q+1}{2})}&\\
&=0.&
\end{aligned}
\end{equation}
\end{proof}

Since both $s$ and $t$ are even,  it is clear that all coordinates of $\overrightarrow{a}_1$ are nonsquares and
all coordinates of $\overrightarrow{a}_2$ are squares. Thus there exists no duplication between these two parts. Choose
$\overrightarrow{a}=(\overrightarrow{a}_1,\overrightarrow{a}_2)$ and $\overrightarrow{v}=(\overrightarrow{v}_1,\overrightarrow{v}_2)$.

\begin{theorem}\label{ThmF}
Let $n=lh+mr$, where even $s\mid q+1$, even $t\mid q-1$, $t\geq 2$, $l=\frac{q^2-1}{s}$, $m=\frac{q^2-1}{t}$, $h\leq \frac{s}{2}$ and $r\leq \frac{t}{2}$.
Then for any $1\leq d\leq min\{\lfloor\frac{s+h}{2}\rfloor\cdot\frac{q+1}{s}-2,\frac{q+1}{2}+\frac{q-1}{t}-1\}$, there exists an
$[[n,n-2d,d+1]]_q$ quantum MDS code.
\end{theorem}

\begin{proof}
The vectors $\overrightarrow{a}$ and $\overrightarrow{v}$ are defined as above. According to Lemmas \ref{thmF3} and \ref{thmF4}, it takes
$$\langle\overrightarrow{a}^{qi+j},\overrightarrow{v}^{q+1}\rangle=\langle\overrightarrow{a}_1^{qi+j},\overrightarrow{v}_1^{q+1}\rangle+ \langle\overrightarrow{a}_2^{qi+j},\overrightarrow{v}_2^{q+1}\rangle=0,$$
for any $0\leq i,j\leq d-1$. Therefore, by Theorem \ref{y2}, the code $\mathbf{GRS}_{d}(\overrightarrow{a},\overrightarrow{v})$ is Hermitian self-orthogonal. By Theorem 1, there exists an $[[n,n-2d,d+1]]_q$ quantum MDS code, where $n=lh+mr$ and
$1\leq d\leq min\{\lfloor\frac{s+h}{2}\rfloor\cdot\frac{q+1}{s}-2,\frac{q+1}{2}+\frac{q-1}{t}-1\}$.
\end{proof}

\begin{remark}\label{R2}
When $h$ approaches to $\frac{s}{2}$ and $t=4$, both $\lfloor\frac{s+h}{2}\rfloor\cdot\frac{q+1}{s}-2$ and $\frac{q+1}{2}+\frac{q-1}{t}-1$
approach to $\frac{3}{4}q$. So the minimum distance of the quantum MDS code can approach to $\frac{3}{4}q$.
\end{remark}

\begin{example}
When $q\equiv 9 \pmod{20}$, applying Theorem \ref{ThmF} with $(s,h,t,r)=(10, 4, 4, 1)$, there exists $q$-ary quantum MDS codes with parameters
\[\left[\left[\frac{13}{20}(q^2-1), \frac{13q^2-28q+79}{20}, \frac{7q-13}{10}\right]\right]_q\]
whose minimal distance is approximately $0.7q$ when $q$ is large.  In general, the length satisfies $\frac{13}{20}(q^2-1)\not\equiv 0,1\pmod{q\pm 1}$.
Therefore this class of quantum MDS codes are new.
\end{example}

\begin{example}
When $q\equiv 29\pmod{60}$, applying Theorem \ref{ThmF} with $(s,h,t,r)=(30, 14, 4, 1)$, there exists quantum MDS codes with parameters
\[\left[\left[\frac{43}{60}(q^2-1), \frac{43q^2-88q+229}{60}, \frac{11q-19}{15}\right]\right]_q\]
whose minimal distance is approximately $11q/15\approx 0.733q$ when $q$ is large. Also the length satisfies  $\frac{43}{60}(q^2-1)\not\equiv 0,1\pmod{q\pm 1}$
and these quantum MDS codes are new.
\end{example}

\section{Conclusion}

 \quad\; Applying Hermitian construction and GRS codes, we construct several new classes of quantum MDS codes over $\mathbb{F}_{q^2}$ through Hermitian
self-orthogonal GRS codes. Some of these quantum MDS codes can have minimum distance bigger than $\frac{q}{2}+1$. Since the lengths are chosen up to two
variables $h$ and $r$. This makes their lengths more flexible than previous constructions. Using our results, we can produce many new quantum MDS codes
with new lengths which have not appeared in previous works. We give an example.

\begin{example}
Choose $q=37$. Utilizing the results in this paper, there are 438 new $[[n,n-2d,d+1]]_{37}$ quantum MDS codes with minimum distance
$d+1\geq \frac{q}{2}+1$, which were not reported in previous papers. We list some of new $[[n,n-2d,d+1]]_{37}$ quantum MDS codes in Table 1.
\end{example}

For a fixed $q$, it is expected to have $[[n,n-2d,d+1]]_q$ quantum MDS codes for any length of $q+1<n\leq q^2+1$ and minimum distance
$\frac{q}{2}+1\leq d+1\leq min\{\frac{n}{2},q+1\}$. But sum up all the results, such quantum MDS codes is still very sparse. It is expected that
more quantum MDS codes with large minimal distance will be explored.

\begin{center}
\begin{longtable}{|c|c|c|}  
\caption{Some of New $[[n,n-2d,d+1]]_{37}$ Quantum MDS Codes} \\ \hline
$n$ & $n-2d$ & $d+1$ \\  \hline
588   &  544   & 23 \\ \hline
624  & 580 & 23 \\ \hline
660  & 614  &  24 \\ \hline
696 &  650 & 24  \\ \hline
702 &  658 & 23    \\ \hline
732 &  684  &  25 \\ \hline
738 & 694 &  23 \\ \hline
768  &  720 &   25 \\ \hline
774 & 728 &   24 \\ \hline
804 & 756 &   25 \\ \hline
810 &  764 &   24 \\ \hline
816 &  772 &  23 \\ \hline
840 & 792 &  25 \\ \hline
846 & 798 &  25 \\ \hline
852 & 808  & 23 \\ \hline
882 & 834   &   25 \\ \hline
918 & 868 &  26 \\ \hline
954 & 904 &  26 \\ \hline
\end{longtable}
 \end{center}

 \section{Acknowledgements}

This research is supported by National Natural Science Foundation of China under Grant 11471008 and Grant 11871025
and the self-determined research funds of CCNU from the colleges' basic research and operation of MOE(Grant No. CCNU18TS028).

\end{document}